\documentclass{article}
\usepackage{amssymb, fullpage}
\usepackage{graphicx, amsmath, amsthm, url}
\newtheorem{lemma}{Lemma}
\newtheorem{fact}[lemma]{Fact}
\newtheorem{proposition}[lemma]{Proposition}
\begin{document}

\title{On the Complexity of the Minimum Latency Scheduling Problem on the Euclidean Plane\thanks{Part of this work was done while both authors were at ITCS, Tsinghua University, Beijing. This work was supported in
part by the National Natural Science Foundation of China Grant
60553001, and the National Basic Research Program of China Grant
2007CB807900, 2007CB807901.}}

\author{Henry Lin \\ University of Maryland 
\\ \url{henrylin@umiacs.umd.edu}
\and Frans Schalekamp \\ unaffiliated \\ \url{fms9@cornell.edu}
}

\maketitle

\begin{abstract}
We show NP-hardness of the minimum latency scheduling (MLS) problem under the physical model of wireless networking.  In this model a transmission is received successfully if the Signal to Interference-plus-Noise Ratio (SINR), is above a given threshold.  In the minimum latency scheduling problem, the goal is to assign a time slot and power level to each transmission, so that all the messages are received successfully, and the number of distinct times slots is minimized.

Despite its seeming simplicity and several previous hardness results for various settings of the minimum latency scheduling problem, it has remained an open question whether or not the minimum latency scheduling problem is NP-hard, when the nodes are placed in the Euclidean plane and arbitrary power levels can be chosen for the transmissions. We resolve this open question for all path loss exponent values $\alpha \geq 3$.
\end{abstract}

\section{Introduction}
In this paper, we give a reduction showing NP-hardness for the minimum latency scheduling (MLS) problem in wireless networking. In the MLS problem, we have nodes in a network, embedded in the Euclidean plane, seeking to send messages to each other in the physical model of wireless communication. For every message a time slot, and a power level has to be chosen. A message is received successfully if the transmission meets a Signal to Interference-plus-Noise Ratio (SINR) requirement. The objective of the problem is to minimize the number of time slots that are used to send all messages successfully. We show that the MLS problem, when arbitrary power levels can be chosen for the transmissions, and the nodes are placed in the Euclidean plane, is NP-hard, for all path loss exponent values $\alpha \geq 3$. This answers the open problem mentioned prominently in a survey paper by~\cite{LRW08}.   The hardness result also shows the MLS problem is NP-hard to approximate within a multiplicative factor better than $\frac{4}{3}$.

The main difficulty to overcome in this reduction, compared to previous hardness results for variants of this problem, is that we have no upper bound on the power levels. This means that in this reduction only the {\em layout} of the nodes (in the Euclidean plane) can be used to make sure a feasible three-round schedule corresponds to a feasible solution of the problem we reduce from.  The key insight to making this idea work, is creating an instance with a large (but polynomial) number of transmissions.

\subsection{Previous Work}
\label{sec:pw}
The physical model was first introduced in a seminal paper by Gupta and Kumar~\cite{GuptaKumar1999}.
There has been much previous work on the problem of scheduling and choosing power levels to send messages in the physical model of wireless networking: in~\cite{GOW07,AD09,KVW10,FKRV09,T09} this problem problems is considered in various settings and under differing assumptions.
In this work, we consider the case when the nodes are placed in the Euclidean plane. This models practical situations in which sensors are deployed in a 2-dimensional area, such as when sensors are deployed in a field to monitor weather conditions, or when sensor are placed along along roads  (without noticeable elevation changes going through valleys or peaks) to monitor traffic conditions.

Instead of trying to describe all different variants for which the MLS problem has been proved to be NP-hard, we will only highlight the result closest to ours, by
Katz, Volker, and Wagner~\cite{KVW10}. For a more complete summary of work in this area, we refer readers to the survey papers~\cite{LRW08,HL09,LP10} and book chapter~\cite{CF09}.

Katz, Volker, and Wagner~\cite{KVW10} show that the MLS problem of scheduling and choosing transmission powers is NP-hard if the network is embedded in the Euclidean plane and there are known upper and lower bounds on the power levels that can be used.  Their proof also shows that if there are only a finite set of transmission powers that can be chosen, then this problem is also NP-hard.  We show that the MLS problem is NP-hard when the nodes are in the Euclidean plane and arbitrary power levels can be chosen.

The distinction between the two problems can be thought of as follows. The work of Katz, Volker, and Wagner~\cite{KVW10}
is applicable in the setting in which we want to schedule messages on sensors that have already been built, as these sensors will naturally have an upper and lower bound on their allowable transmission powers due to hardware or software limitations.
The setting we consider is more general than that, and also includes the case in which
it is possible for the system designer to choose the sensors and the power levels at which they can transmit, so as to maximize the number of transmissions that can be sent in a given time period, and/or minimize the total number of time periods that are needed to send a certain set of messages.  Despite its seeming simplicity, the hardness of this problem has remained open for many years, and was even featured prominently as an open question in a survey by Locher, Rickenbach, and Wattenhofer~\cite{LRW08}.

The best algorithm for the MLS problem with arbitrary power levels is due to Kesselheim~\cite{K11}, who presents an $O( \log^2 n )$-approximation algorithm for the MLS problem in general metrics, and an $O( \log n )$-approximation for {\em fading metrics}. The notion of {\em fading metrics} was introduced by Halld\'orsson~\cite{H09}, and defined to be the combination of a metric and a path loss exponent, so that the doubling dimension of the metric is strictly less than the path loss exponent $\alpha$. It was introduced as a generalization of the combination of the metric on the Euclidean plane, and a path loss exponent $\alpha$ that is strictly larger than $2$. Note that the result presented here, is for this ``easier'' variant.

\subsection{Outline}
The outline of this paper is as follows. We define the MLS problem formally, and describe the physical model (SINR model) in Section~\ref{sec:pre}. After a high level overview of the reduction in Section~\ref{sec:highlevel}, we give a detailed description of the reduction and the proof the main result of this paper, namely that the MLS problem is NP-hard, when the nodes are placed in the Euclidean plane and arbitrary power levels can be chosen for the transmissions, in Section~\ref{sec:mainresult}. We conclude with open problems in Section~\ref{sec:openproblems}.

\section{Preliminaries}
\label{sec:pre}
\subsection{Minimum Latency Scheduling Problem}
In the minimum latency scheduling problem, we are given a set of nodes $X$ in a metric space and a distance function $d: X \times X \to R$, defining distances between any pair of points in the metric space.  In this paper, we will consider the special case where the set of nodes $X$ are on the Euclidean plane, and the distance function is the 2-dimensional Euclidean distance metric $d( x, y ) = || x -  y ||_2$.  The set of messages we wish to transmit will be represented by a multiset $M$ over $X \times X$.  In the minimum latency scheduling problem, we seek to partition the multiset $M$ into $T$ sets $M_1, M_2, \dots, M_T$ to be transmitted over $T$ rounds (or time slots), and choose power levels for the messages to be transmitted in each round, such that all messages are successfully transmitted according to the SINR model.  The goal will be to minimize $T$, the number of rounds needed to transmit all the messages successfully.

\subsection{Physical Model (SINR Model)}
In the SINR model, we have a parameter $\alpha$ called the path loss exponent, which determines how quickly the signal deteriorates as the distance increases. It is often assumed that $\alpha > 2$ as it is in real settings.  If a node $u$ is transmitting a message to node $v$ in some round $r$, we define the \emph{signal} received by node $v$ in round $r$ to be $S_r(v) = P_r(u)/d(u,v)^\alpha$, where $P_r(u)$ is the power used by node $u$ to transmit its message in round $r$.  When node $u$ is transmitting a message to a node $v$ in round $r$, we also define the \emph{interference} of node $v$ in round $r$ to be $I_r(v) = \sum_{w \in V \setminus \{ u \} } P_r(w)/d(v,w)^\alpha$, where $P_r(w) = 0$ if a node $w$ is not transmitting anything in a round $r$.  Lastly, each node $v$ has a noise parameter $n_v$ that determines how much background noise it experiences, and a message to node $v$ succeeds if the signal to interference plus noise ratio exceeds some threshold $\beta_v \geq 1$.  More formally, a message from node $u$ to $v$ succeeds in round $r$, if the following SINR inequality is satisfied:

\[ \frac{S_r(v)}{I_r(v) + n_v} = \frac{P_r(u)/d(u,v)^\alpha}{\sum_{w \in V \setminus \{ u \} } P_r(w)/d(v,w)^\alpha + n_v} > \beta_v. \]

\subsection{No Noise, Unit Threshold}
In the reduction presented, we will set $n_v = 0$ for all nodes $v$, thus proving that this problem is NP-hard for nodes placed in the Euclidean plane, {\it even if there is no background noise}. Note that this ``no noise'' assumption only potentially makes the {\it reduction} harder, and proving that problem is NP-hard without noise, also shows the problem with noise is NP-hard.  Additionally, we will set $\beta_v = 1$ for all $v$, which shows this problem is NP-hard even when all SINR threshold values are set to 1.  In the special case that $n_v = 0$ and $\beta_v = 1$ for all nodes $v$, the SINR inequality simply requires that $S_r(v) > I_r(v)$ for a message to be successfully transmitted.

\section{High Level Overview of the Reduction}
\label{sec:highlevel}
We show that the minimum latency scheduling problem is NP-hard, by reducing the 3-coloring problem of planar 4-regular graphs to it, which was shown to be NP-hard by Dailey~\cite{Dailey1980289}.  In the 3-coloring problem on planar 4-regular graphs, we are given a planar graph $G$ of degree 4, and we want to determine whether there is a proper coloring of the nodes of $G$ that uses only 3 colors.  We can assume that the graph is given as an orthogonal grid graph, which is graph drawn in the plane with all nodes placed at grid intersection points, and all edges drawn along grid lines without any edge crossings, but edges may bend at grid intersection points.   This coloring problem on an orthogonal grid graph is also NP-hard because Shiloach~\cite{Shiloach} and Valiant~\cite{Valiant} show that any planar $G$ of degree at most 4 can be drawn as an orthogonal grid graph in an area of size at most $O( |V|^2 )$.

We construct an MLS instance whose transmissions can be completely scheduled in 3 rounds if and only if $G$ is 3-colorable.
For each node $v$ in $G$, we will build a \emph{node gadget} consisting of two rows of nodes $N_v$ and a set of messages $M_v$ within $N_v$, which can be decomposed into 3 disjoint sets of messages $R_v$, $B_v$, and $G_v$ (where the names of the sets are chosen so as to suggested color classes (red, blue and green)), such that $M_v = R_v \cup B_v \cup G_v$.  The construction will be guaranteed to have the property that any 3 round solution for MLS instance must exactly transmit one of the sets $R_v$, $B_v$, or $G_v$ for each of the 3 rounds of the solution.  We call this the \emph{node color consistency property},  because if the messages for node $v$ are transmitted as the sets $R_v$, $B_v$, and $G_v$ in 3 rounds, then we can use these message transmissions to define a color for $v$ in each of the three rounds.

To guarantee that the three coloring will be a proper coloring where each pair of adjacent nodes has different colors, we construct three \emph{edge gadgets} for each edge $\{u,v\}$ in the graph consisting of the messages $R_{\{u,v\}}$, $B_{\{u,v\}}$, and $G_{\{u,v\}}$.  The edge gadget consisting of the messages $R_{\{u,v\}}$ will have the property that they cannot be transmitted in 3 rounds if the set of messages $R_u$ and $R_v$ are ever transmitted in the same round.  Similarly, the set of messages $B_{\{u,v\}}$ and $G_{\{u,v\}}$ will have the property that they cannot be transmitted in 3 rounds if the set of messages $B_u$ and $B_v$, or $G_u$ and $G_v$, are ever transmitted in the same round, respectively.  We call satisfying these requirements the \emph{edge color constraint property}.  Together with the node color consistency property mentioned in the last paragraph, this edge color constraint property, will ensure that if all messages are successfully transmitted in 3 rounds, then each round of messages transmitted will define a coloring for the nodes in the original graph $G$, and that coloring must have the property that adjacent nodes are colored different colors.

\section{Reduction}
\label{sec:mainresult}
We will now describe the reduction in detail, by first describing the node and edge gadgets, and then explaining the way the gadgets are connected.

\subsection{Node Gadget}

\begin{figure}
\begin{center}
\includegraphics[width=0.8\textwidth]{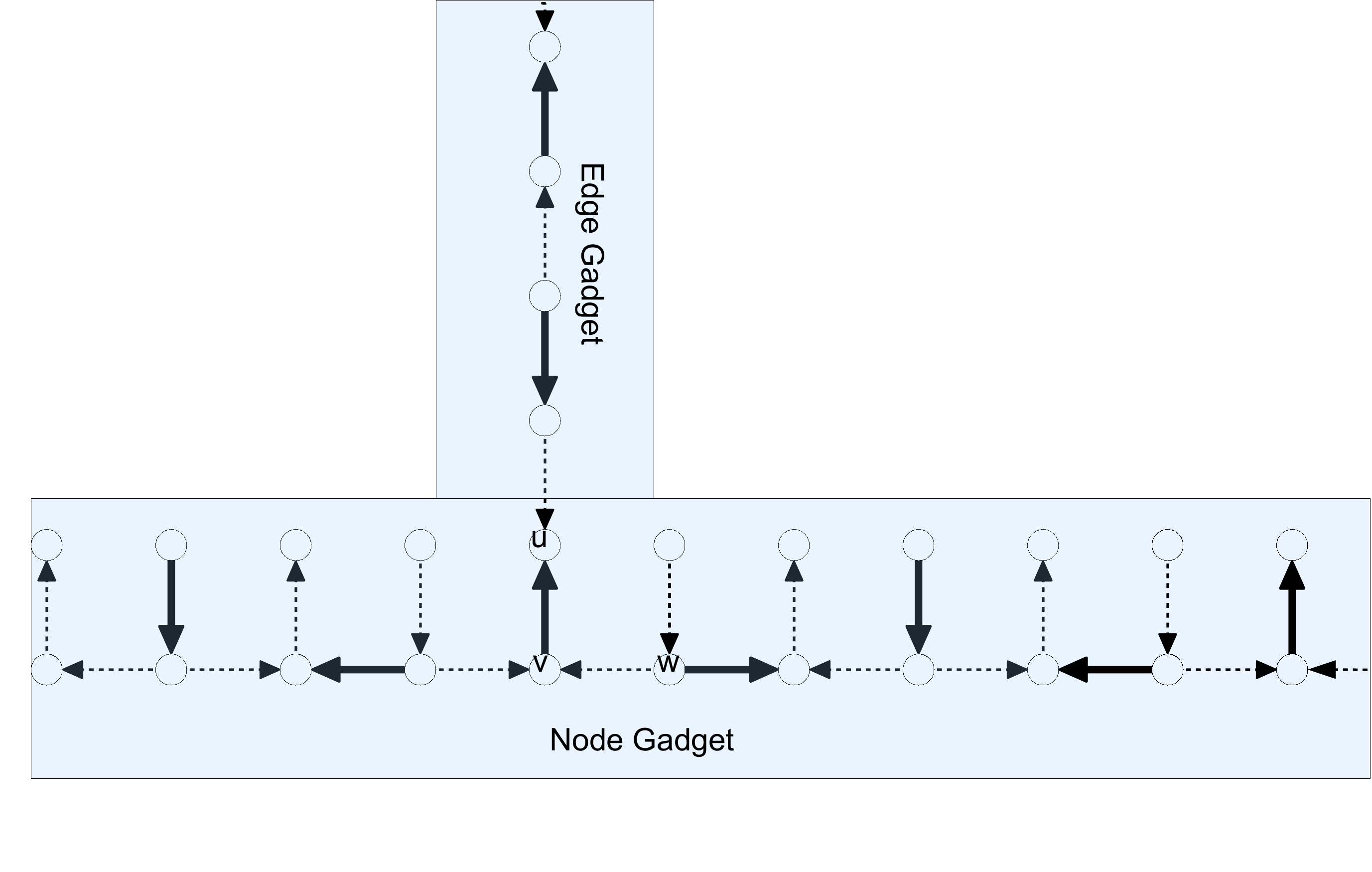}
\end{center}
\caption[Node gadget]{An illustration of the construction of a node gadget and how it is connected to an edge gadget.  The messages in bold illustrate the messages in $R_u$.}
\label{fig:nodegadgeton}
\end{figure}
\begin{figure}
\begin{center}
\includegraphics[width=.7\textwidth]{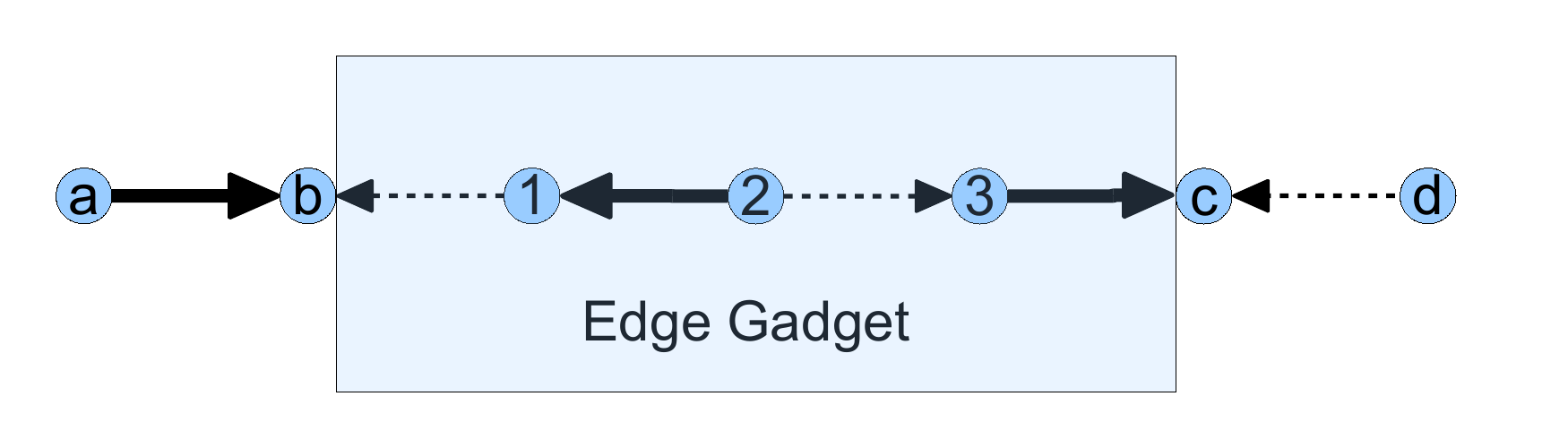}
\end{center}
\caption[Edge gadget]{An illustration of an edge gadget enforcing that the messages $(a,b) \in M_u$ and $(d,c) \in M_v$ cannot transmitted in the same round, if all transmissions are to finish in 3 rounds. Each bold arrow represents a message to be transmitted once, while each dashed arrow represents a message to be transmitted twice.}
\label{fig:edgegadget}
\end{figure}
\begin{figure}[t!]
\begin{center}
\includegraphics[width=.35\textwidth]{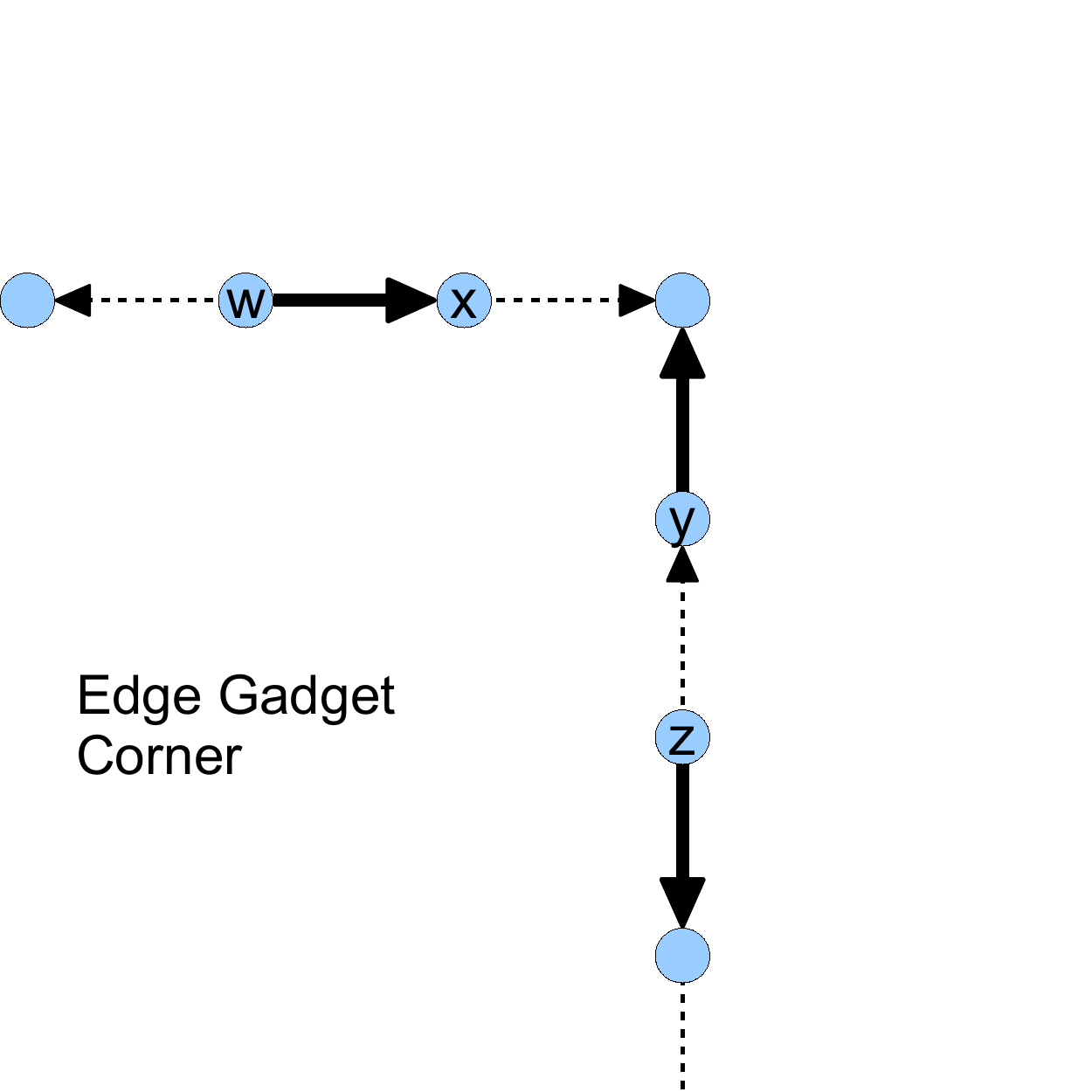}
\end{center}
\caption[Edge gadget]{An illustration showing how a turn in an edge gadget can be achieved.}
\label{fig:edgeturn}
\end{figure}

The node gadget consists of two rows of nodes and transmissions between them, according to the pattern sketched in Figure~\ref{fig:nodegadgeton}.  Transmissions are in indicated by arrows in the figure.  Each transmission is unit length, and the nodes in each node gadget can be assumed to lie on integer grid points. The total length of the node gadget is $\Theta( n^2 )$, where $n$ is the number of nodes in the coloring instance we are reducing from.  The full node gadget consists of the pattern shown in Figure~\ref{fig:nodegadgeton} repeated many times to the left and to the right of the picture.  As the pattern is repeated, each node in the bottom row of the node gadget has three transmissions in which it participates as either a sender or receiver (except for the nodes on the left most and right most end of the gadget, which only have two transmissions).  Along the length of the node gadget up to 12 edge gadgets (3 edge gadgets for each possible edge) may be connected to the node gadget, although we will make sure the edge gadgets are adequately spaced by distance $\Theta( n^2 )$ to ensure that the interference can still be bounded easily.

The transmissions of the node gadget corresponding to node $u$ in the coloring instance, are grouped into three sets, $R_u$, $B_u$ and $G_u$. In Figure~\ref{fig:nodegadgeton}, the transmissions in $R_u$ are highlighted by the darker solid edges.  Ignoring the directionality of the edges, the pattern of the transmissions for $G_u$ and $B_u$ are exactly the same, except that the pattern is shifted to the right by one node for the transmissions in $G_u$, and shifted right by two nodes for $B_u$. In the Appendix, we prove the ``node color consistency property''.

\begin{lemma}[Node color consistency]
If there exists a three round solution to the constructed MLS instance then the messages sent in any round are exactly the messages of one of the sets $R_u$, $G_u$ or $B_u$.
\end{lemma}

\subsection{Edge Gadget}
Each edge $\{u,v\}$ in the coloring instance will give rise to three edge gadgets in the MLS instance, consisting of the messages $R_{\{u,v\}}$, $B_{\{u,v\}}$, and $G_{\{u,v\}}$, to ensure that the nodes $u$ and $v$ cannot transmit messages of the same color, if a 3 round MLS solution is to be achieved.  The $R_{\{u,v\}}$ messages will share a node in common with one of the node gadget messages in $R_u$ and one of the messages in $R_v$ to ensure that the messages in $R_u$ and $R_v$ cannot be transmitted at the same time in any 3 round solution for the constructed MLS problem.  Similarly, $B_{\{u,v\}}$ will share a node in common with one message in $B_u$ and one message in $B_v$, and $G_{\{u,v\}}$ will share a node in common with one message in $G_u$ and one message in $G_v$, in order to maintain the edge color constraint property.

In the construction, the node and edge gadgets will be spaced fairly far apart at a distance of $\Theta( n^2 )$ to make bounding the interference of any transmission easier.  However, each message constructed will either have unit length or a slightly longer length of 1.1, so that the signal between nodes remains adequately high.  To illustrate the construction of the edge gadget, imagine that we have the simpler goal of ensuring that the messages $(a,b)\in R_u$ and $(d,c)\in R_v$ shown in Figure~\ref{fig:edgegadget} are not transmitted in the same round in any 3 round MLS solution.  We can maintain this property by simply adding the pattern of messages shown in Figure~\ref{fig:edgegadget}, where each bold arrow in the figure represents a message to be sent once and each dashed arrow represents a message to be sent twice.  For this \emph{mini edge gadget}, we can show the following lemma holds:

\begin{lemma}[Edge color constraint lemma]
Given the mini edge gadget construction shown in Figure~\ref{fig:edgegadget}, there exists no three round solution to the constructed MLS instance in which the $(a,b)$ and $(d,c)$ are sent in the same round.
\end{lemma}
\begin{proof}
Each node in the edge gadget has 3 transmissions to send or receive, so each of those nodes must be either sending or receiving a message in each round, if they are to complete all their transmissions in 3 rounds, since the SINR model disallows any node from both sending and receiving a transmission at the same time.  If $(a,b)$ and $(d,c)$ are both being transmitted in one round that leaves 3 edge gadget nodes, who each must pair with each other to send or receive a transmission.  Since there are an odd number of edge gadget nodes looking to send or receive a message, there must be at least one edge gadget node which cannot send or receive a message in the round where both $(a,b)$ and $(d,c)$ are sent, implying that not all edge gadget messages can be sent in 3 rounds.
\end{proof}

To connect node gadgets which may be spaced arbitrarily far apart, note that this pattern of 4 edges (consisting of 6 transmissions) can be repeated as many times as needed, by repeating copies of the pattern next to each other, while matching node $b$ of one copy of the pattern to node $c$ of the next copy of the pattern.  Note that the edge color constraint property still holds for repeated copies of the mini edge gadget, since if the boundary node messages $(a,b)$ and $(d,c)$ are transmitted at the same time, then there still remains an odd number of nodes, which all must transmit or receive a message in the round, which is not possible.  Additionally, since the edge gadgets will need to make turns in order to connect to node gadgets in the proper orientation, we will construct turns in the edge gadget as shown in Figure~\ref{fig:edgeturn}.   A turn consists of the pattern of edges shown in Figure~\ref{fig:edgeturn} (constructed by connecting two mini edge gadgets at the corner node), and can also be repeated as many times as needed (possibly rotated 90, 180, or 270 degrees) within the edge gadget, while still maintaining the edge color constraint property: if the top left node and the bottom right node are receiving a message that does not belong to the turn pattern in the same round, then there are again an odd number of nodes (5) which all must transmit or receive a message in that round, which is not possible.

Lastly, note that since the edge gadget always consists of a repeated pattern of 4 edges, if we only use edges of unit length, then the edge gadgets will always have a length which is a multiple of 4, while we may need to build edge gadgets of an integer length which is not a multiple of 4 to connect node gadgets.  For this reason, we make use of two different transmissions lengths, 1 and 1.1, to give enough flexibility to connect any two node gadgets with an edge gadget of any integer length greater than 40.  By replacing 0, 10, 20, or 30 transmissions of unit length transmissions with length $1.1$ transmissions, then we create edge gadgets of any integer length greater than 40, which may be necessary for the node gadget connections described in the next section.

\subsection{Node Gadget and Edge Gadget Connections}

Now that we have described the edge gadgets and the node gadgets, we will now describe how the edge gadgets are connected between different node gadgets.  Note that each node in the original coloring instance can have up to 4 edges, so the node gadget must be long enough to have connections to up to 12 edge gadgets, 3 for each edge of the coloring instance.  In Figure~\ref{fig:supernodegadget}, we show how an node gadget is connected to 4 other node gadgets to the left, right, above, and below it.  Note that this can be done without crossing any of the edge gadgets, even though all the edge gadgets need to be connected to the ``top'' part of the node gadget.  Additionally, to ensure no edge gadgets cross outside of the figure, we will use the placement of nodes and edges in the original planar grid coloring instance to help guide the placement of nodes and edges, such that no edge gadgets cross.

\begin{figure}
\begin{center}
\includegraphics[width=.7\textwidth]{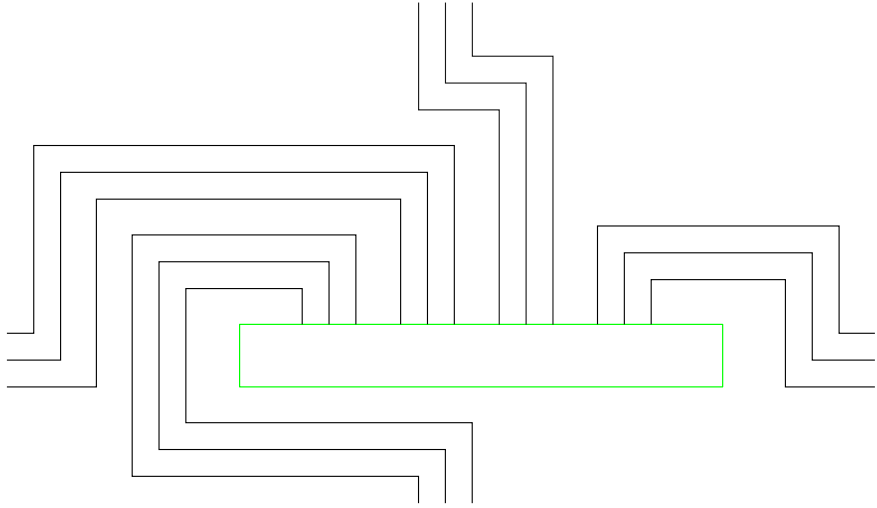}
\end{center}
\caption[Sketch of a node gadget and $12$ edge gadgets]{Sketch of the interaction between a node gadget and its edge gadgets: The green box indicates a node gadget. Every line represents an edge gadget, which represents the combination of an edge in the coloring instance and a color.}
\label{fig:supernodegadget}
\end{figure}

Moreover, we can ensure that no edge gadgets cross, even while spacing out the node gadgets and making the node gadgets have length $\Theta(n^2)$, so that the edge gadgets are spread apart by distance by distance at least $3 n^2$, where $n$ is the number of nodes in the original coloring instance.  This will ensure that when we consider a ball of radius $n^2$ around any node, it will intersect with at most one node gadget and one edge gadget, and this will make bounding the interference of any transmission easier later.

We can place all the nodes and edges to match the requirements by taking the placement of nodes in the original coloring instance, and placing the node gadgets in a larger grid graph, which has dimensions which are a factor of $\Theta(n^2)$ larger and the node gadgets are spaced out by a multiplicative factor of $\Theta(n^2)$ more than they were before.  After spacing the node gadgets out, the edge gadgets can follow roughly the same path as laid out by the original coloring instance, all while maintaining the distance of $3 n^2$ from any other node gadget.  Additionally, we can make sure all node gadget nodes lie on integer grid points, and all edge gadget connections have integer length greater than 40, so that they can be properly constructed by repeating the patterns shown in Figure~\ref{fig:edgegadget} and Figure~\ref{fig:edgeturn}.  We can also make sure that the turns occurring in the edge gadgets are also spaced out by a distance of $3 n^2$ as well, so we do not have any edge turns closer than $n^2$ from each other.

Even though we have increased the number of nodes and messages in the constructed MLS instance, we can still bound the total number of nodes and messages in the MLS instance by $O(n^4)$, which will be useful later for bound the interference to each node.  This is because we can assume that the coloring instance with $n$ nodes is drawn on a grid with area at most $O(n^2)$, and each node in the original coloring instance is replaced by a node gadget of size $\Theta(n^2)$.  Furthermore, the total length of all edges in the original grid is at most $O(n^2)$, and each unit length of edge is transformed into $\Theta(n^2)$ edges by the reduction, so the number of nodes in the node and edge gadgets are both bounded by $O(n^4)$.

\begin{proposition}
The constructed minimum latency scheduling problem has a 3 round solution if and only if the original planar grid graph is 3-colorable.
\end{proposition}

\begin{proof}
Note that the previous proof of the node color consistency property and the edge color constraint property have already proved that a 3 round solution for the MLS problem implies a 3 coloring for the coloring instance, since the color of the messages that are transmitted in the first round of the MLS solution must be a valid coloring for the original coloring instance.  To complete the reduction, we just need to show that a 3 coloring for the original grid graph implies a 3 round solution for the constructed MLS problem.  For this part of the proof, we will assume that $\alpha$ has some value $\geq 3$, which will be necessary for ensuring that all transmissions can be feasibly scheduled in 3 rounds, when a feasible 3 coloring exists.

For the proof, we will assume that we are given a coloring of original coloring instance consisting of the colors 1, 2 and 3, and we will use it define the transmissions to be used in each of the 3 rounds of the constructed MLS instance.  Given a coloring for the original graph, we will use the following transmission scheme: in the first round we send the $R_u$ transmissions in the node gadgets corresponding to nodes $u$ that are colored with color 1, $B_v$ for the nodes gadgets which correspond to nodes $v$ colored with color 2, and $G_w$ for the node gadgets corresponding to nodes $w$ that are colored with color 3.  In the second round, we rotate the color transmissions, so that we send the $B_u$ transmissions for the nodes $u$ that are colored with color 1, $G_v$ for the nodes $v$ that are colored with color 2, and $R_w$ for the nodes $w$ that are colored with color 3.  Finally in the third round, we transmit all remaining messages.

Now that we have defined all the transmissions for all the node gadgets, we also have to define the transmissions for the edge gadgets.  To schedule the messages that occur in the edge gadgets, consider the basic building block of an edge gadget shown in Figure~\ref{fig:edgegadget}.  The messages shown in bold in Figure~\ref{fig:edgegadget} are scheduled precisely when the message $(a,b)$ is transmitted in Figure~\ref{fig:edgegadget}, and the other messages shown as dashed lines are transmitted in the other 2 rounds.  This same transmission pattern holds for the full edge gadget, which repeats many copies of the building block shown in Figure~\ref{fig:edgegadget}.

The power level for each transmission is 1, with the exception of messages that are at the ``boundary'' of an edge gadget transmitting to an node gadget node, as shown in Figure~\ref{fig:nodegadget} in the Appendix, where node $u$ receives a message from the edge gadget node above it.  These ``boundary'' messages are transmitted at power 2.

The calculations that show all the transmissions we defined are successful is given in the Appendix, by bounding the interference to a few key nodes, and then showing that the bounds provided can also be used to bound the interference to all other receiving nodes, to be below the signal they receive.  The main idea will be to first show that all nodes whose distance is greater than $n^2$ units away contribute at most $o(1/n)$ interference, which is negligible and can be ignored for the purposes of the analysis.  After ignoring transmissions further than $n^2$ units away, we more precisely upper bound the interference of transmissions which occur within $n^2$ distance of the node we are analyzing.  This final interference calculation completes the proof of the reduction by showing a 3 coloring of the original graph implies a 3 round solution for the constructed MLS instance.
\end{proof}

\section{Open Problems}
\label{sec:openproblems}
An interesting open question is whether there exists an $O( 1 )$-approximation algorithm for this problem.  Alternatively, it would be interesting to improve this hardness result to show that the MLS problem is hard to approximate better than $O( \log n)$ to match the best known approximation algorithm.  As we mentioned previously, the best known approximation algorithm is due to Kesselheim~\cite{K11}, who presents an $O( \log^2 n )$-approximation algorithm for the MLS problem in general metrics, and an $O( \log n )$-approximation for fading metrics.  In fact, we do not even know whether it is harder to approximate this problem without power restrictions, than when power levels are limited.

Since we show that our proof holds for $\alpha \geq 3$ in the appendix, it is an interesting open question to show NP-hardness for $\alpha$ arbitrarily close to 2, and without any restrictions on the power levels used. For $\alpha$ arbitrarily close to 2, we are not sure whether the reduction presented here can be used.  We have an NP-hardness result for the case where the parameter $\alpha$ can have any arbitrary value, but then we need the additional assumption that there is a maximum power level for the transmissions, and in the reduction we make use of the noise parameter.  This result is also stronger than what is known previously, since~\cite{KVW10} for example, requires both an upper {\it and} lower bound on the power levels.  Finally, we would like to note that the result here holds for centralized algorithms, and therefore also for distributed algorithms.  Another important problem is to find a good distributed algorithm for this problem.

\bibliographystyle{plain}
\bibliography{wireless}

\appendix

\section{Proof of Node Color Consistency}\label{sec:colorc}

We will call the transmissions which go from a node in the upper row of the node gadget to a node in the lower row of the node gadget \emph{down} transmissions, while edges which go from a node in the lower row to a node in the upper row will be called \emph{up} transmissions.  We will use the term \emph{vertical} transmissions to collectively refer to both the up and down transmissions within a node gadget.  Lastly, the transmissions between two nodes in the lower row of the node gadget will be called \emph{horizontal} transmissions.

\begin{lemma}[Node color consistency]
If there exists a three round solution to the constructed MLS instance then the messages sent in any round are exactly the messages of one of the sets $R_u$, $G_u$ or $B_u$.
\end{lemma}
\begin{proof}
Since each internal node in the lower row of the node gadget is participating in exactly 3 transmissions, it must be the case that is participating in exactly 1 transmission each round in the three round solution.

Consider two neighboring vertical transmissions (i.e., the bottom nodes participating in these transmissions are 1 unit apart), say $t_1$ and $t_2$. Note that the distance between the sender of $t_1$ and the receiver of $t_2$ is $1$, the same as the length of the transmissions. Therefore, these transmissions cannot take place in the same round, by Lemma~\ref{crossinglemma}, which we prove now.

\begin{lemma}\label{crossinglemma}
Suppose there are two transmissions, $i_1$ to $i_2$ and $j_1$ to $j_2$, such that the $d( i_1, j_2 ) \leq d( i_1, i_2 )$ and $d( j_1, i_2 ) \leq d( j_1, j_2 )$.  Then these two transmissions cannot take place in one round in the SINR model.
\end{lemma}
\begin{proof}
Assume by means of contradiction that we have a feasible transmission schedule so that these transmissions take place in the same round.
Let $P_1$ be the power level of transmission $i_1$ to $i_2$, and $P_2$ be the power level of transmission $j_1$ to $j_2$. We get the following string of inequalities:
\[
	P_1 / d( i_1, i_2 )^\alpha  >  P_2 / d( j_1, i_2 )^\alpha \geq P_2 / d( j_1, j_2 )^\alpha  > P_1 / d( i_1, j_2 )^\alpha  \geq P_1 / d( i_1, i_2 )^\alpha,
\]
The inequalities follow because the signal for the messages has to be larger than the interference, and due to the assumption on the distances. We thus obtain $P_1 > P_1$, which is obviously not possible.
\end{proof}

Now consider two vertical transmissions where the bottom nodes participating in these transmissions are 2 units apart, say $t_1$ and $t_2$. We know by the previous reasoning that the transmission between these two transmissions, say $t_3$, cannot occur in the same round. But the bottom node of $t_3$ has to participate in a transmission every round, and the other two transmissions it is part of cannot occur either in this round, because it involves nodes that participate in $t_1$ and $t_2$. Therefore, in a feasible three round solution, $t_1$ and $t_2$ cannot be sent in the same round.

Consider a three round solution where, in one round, there are two vertical transmissions where the bottom nodes participating in these transmissions are more than 3 units apart, and there is no vertical transmission between these transmissions in that round. This implies that a sequence of 3 adjacent vertical messages that can be sent in the other two rounds. But this is not possible without violating the fact that adjacent vertical messages can never be sent in the same round, and vertical messages with 2 units of horizontal distance between them cannot be sent either, as we argued previously.

Therefore, the vertical messages must be sent in a pattern where the vertical messages sent are exactly 3 units apart in each round, and color consistency must hold for each node gadget in any feasible 3 round solution.
\end{proof}

\section{Proof that all messages can be successfully transmitted}

\begin{figure}
\begin{center}
\includegraphics[width=0.8\textwidth]{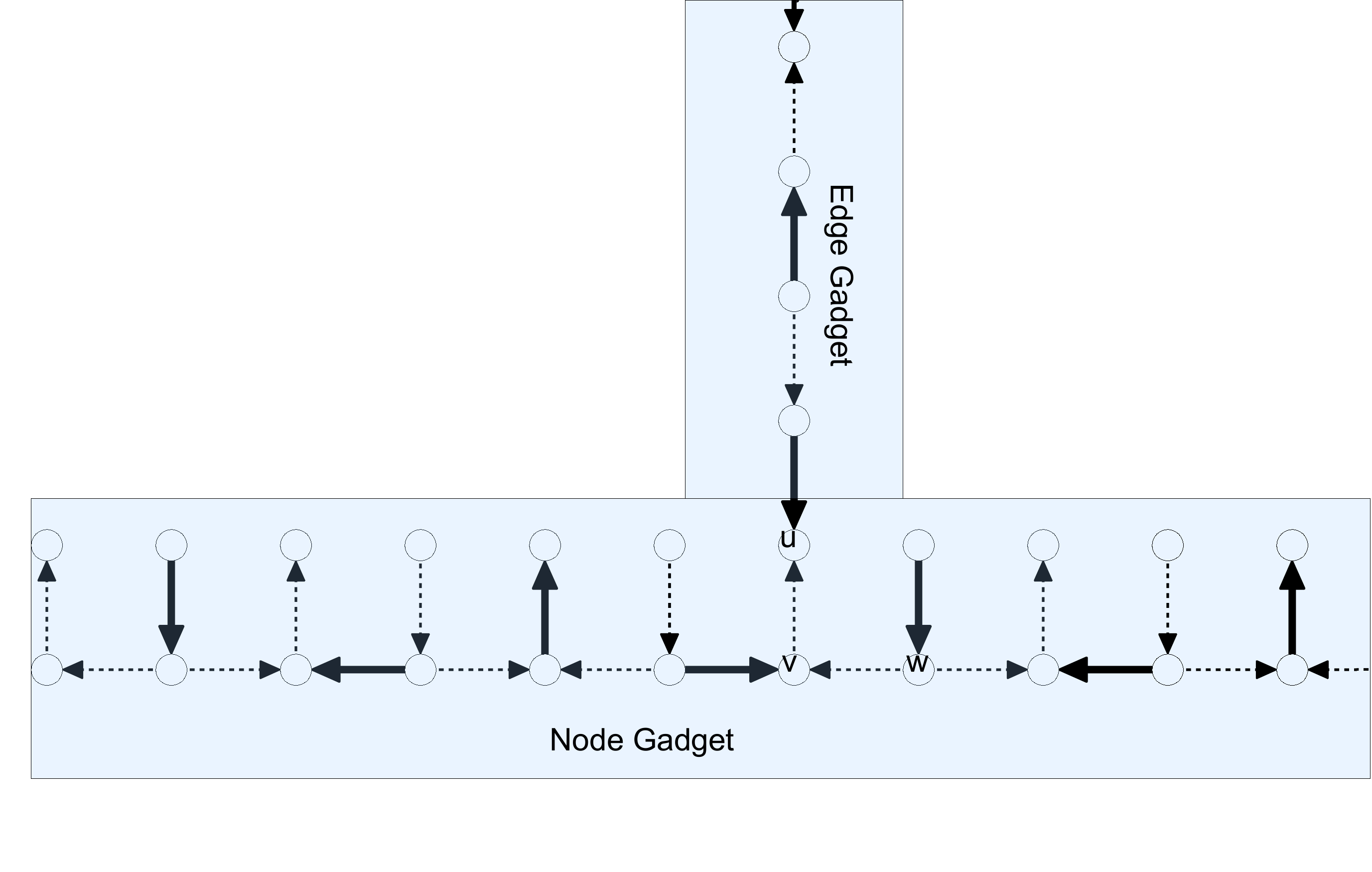}
\end{center}
\caption[Node gadget]{An illustration of a node gadget and its connection to an edge gadget.  The messages in bold illustrate another transmission pattern in which we need to bound the interferences, and show the message transmissions are successful.}
\label{fig:nodegadget}
\end{figure}

In this section we bound the interference to all receivers in the edge and node gadgets to complete the proof of correctness for the reduction.

\begin{fact}
The interference for any transmission of transmissions whose senders are a distance greater than $n^2$ units away is $o( 1/n )$.
\end{fact}
\begin{proof}
There are at most $O(n^4)$ total nodes in the construction, and each node further away than $n^2$, contributes at most $2 n^{-2\alpha} = O(n^{-6})$ interference when $\alpha\geq 3$, since every power level is at most $2$.  Thus the total interference from these far away nodes is at most $O(n^4) \times O(n^{-6}) = o(1/n)$.
\end{proof}

For the purposes of the analysis, it may also be useful to note that the messages transmitting at power 2 are only below edge gadgets, which are separated by a distance of at least $3 n^2$.  This means that when we are bounding the interference of nodes within distance $n^2$, we only need to explicitly bound the interference of one node transmitting at power 2, while assuming that all other nodes transmit at power 1.

We start by bounding the interference from transmissions within a distance of $n^2$, to the nodes $v$, $u$, and $w$ in the configuration shown in Figure~\ref{fig:nodegadget}, which are in some sense ``worst case'' examples, as we will make clear shortly.

\begin{fact}\label{fact:v}
The total interference from transmissions within a distance of $n^2$, to node $v$ in Figure~\ref{fig:nodegadget} is bounded by $0.94$.
\end{fact}
\begin{proof}
Within the node gadget, the interference caused by nodes in the node gadget to the left of the $v$ can be upper bounded by $2^{-\alpha} + 3^{-\alpha} + \sum_{i=5}^{\infty} i^{-\alpha}$, while the interference caused by nodes in the node gadget to the right of $v$ can be upper bounded by $\sqrt{2}^{-\alpha} + 3^{-\alpha} + 4^{-\alpha} + \sum_{i=5}^{\infty} i^{-\alpha}$.  Lastly, the interference by nodes in the edge gadget above node $v$ can be upper bounded by $2 \cdot 2^{-\alpha} + 3^{-\alpha} + \sum_{i=5}^{\infty} i^{-\alpha}$.  To bound the term $\sum_{i=5}^{\infty} i^{-\alpha}$, note that $\sum_{i=5}^{\infty} i^{-\alpha} = 5^{-\alpha} + \sum_{i=6}^{\infty} i^{-\alpha} \leq 5^{-\alpha} + \int_{i=5}^{\infty} i^{-\alpha}$ since $i^{-\alpha}$ is nonincreasing. We thus get $\sum_{i=5}^{\infty} i^{-\alpha} \leq \frac{1}{\alpha-1} \frac{1}{5^{\alpha-1}} + 5^{-\alpha}$.  Thus, we can bound the total interference to node $v$ by $2^{-\alpha} + 3^{-\alpha} + \sqrt{2}^{-\alpha} + 3^{-\alpha} + 4^{-\alpha} + 2 \cdot 2^{-\alpha} + 3^{-\alpha} + 3\cdot(\frac{1}{\alpha-1} \frac{1}{5^{\alpha-1}} + 5^{-\alpha}) \leq 0.94$, when $\alpha = 3$.
\end{proof}

\begin{fact}\label{fact:boundary}
The total interference from transmissions within a distance of $n^2$, to node $u$ in Figure~\ref{fig:nodegadget} is bounded by $1.73$.
\end{fact}
\begin{proof}
By applying the same analysis, the total interference to node $u$ is upper bounded by $\sqrt{2}^{-\alpha} + (\sqrt{5})^{-\alpha} + (\sqrt{10})^{-\alpha} + 1 + (\sqrt{10})^{-\alpha} + 4^{-\alpha} + 2^{-\alpha} + 3 \cdot \sum_{i=5}^{\infty} i^{-\alpha} \leq 1.73$, when $\alpha = 3$.  The first 3 terms bound the interferences of the 3 nodes to the left of $u$, and the next three terms bound the interferences of the 3 nodes to the right of $u$.  The remaining two terms bound the interferences of all the remaining nodes in the edge gadget within distance $n^2$, and the nodes of distance greater than 4 in the node gadget of $u$.
\end{proof}

\begin{fact}\label{fact:w}
The total interference from transmissions within a distance of $n^2$, to node $w$ in Figure~\ref{fig:nodegadget} is bounded by $.65$.
\end{fact}
\begin{proof}
By similar reasoning as in the previous proofs, we bound the interference to node $w$ by $2^{-\alpha} + 3^{-\alpha} + 4^{-\alpha} + 2^{-\alpha} + 3^{-\alpha} + 4^{-\alpha} + 2 \cdot \sqrt{5}^{-\alpha} + 3^{-\alpha} + 3 \cdot \sum_{i=5}^{\infty} i^{-\alpha} \leq .65$.
\end{proof}

\begin{fact}\label{fact:u}
The total interference from transmissions within a distance of $n^2$, to node $u$ in Figure~\ref{fig:nodegadgeton} is bounded by $.9994$.
\end{fact}
\begin{proof}
Here we have to be very precise about the analysis, and compute the interference exactly for all nodes within distance 10, and then upper bound the interferences from nodes whose interference are greater than distance 10 with integrals.  In the configuration shown, within distance 10 of $u$, there are nodes transmitting to the left of $u$ at distances $\sqrt{2}, 3, \sqrt{26}, \sqrt{37}, \sqrt{50}, 9$, and the interference to the left within the node gadget of $u$ can be bounded by $\sum_{i=11}^{\infty} i^{-\alpha} \leq \int_{i=11}^{\infty} i^{-\alpha} + 11^{-\alpha} \leq \frac{1}{\alpha-1} \frac{1}{11^{\alpha-1}} + 11^{-\alpha}$.  These two exact terms also bound the total interference from nodes to the right of $u$ within $u$'s node gadget.  Additionally, we can bound the interference of all the messages being transmitted in the edge gadget above $u$ and within $n^2$, by exactly computing the interference of the nodes occurring at distances $2, 3, 6, 7, 10$, and again using $\frac{1}{\alpha-1} \frac{1}{11^{\alpha-1}} + 11^{-\alpha}$ to bound all interferences in the edge gadget of distance greater than 10 and within distance $n^2$.  Numerically computing and summing the terms above, we get that the total interference to node $u$ is upper bounded by $0.9994$, which is just slightly less than the signal of 1 that it receives from its sender.
\end{proof}

\begin{fact}\label{fact:allnode}
The total interference from transmissions within a distance of $n^2$, to any message in a node gadget is strictly less than their signal level $1$.
\end{fact}
\begin{proof}
For an up message, the claim follows from Fact~\ref{fact:u}, since this is the worst case with an edge gadget directly above it and a maximal number of messages that can occur to the left and right of it. Similarly, for a down message, the claim follows from Fact~\ref{fact:w}, since this is the worst case with an edge gadget directly above it and a maximal number of messages that can occur to the left and right of it. Finally, for a vertical message, the claim follows from Fact~\ref{fact:v}.
\end{proof}

\begin{fact}
All nodes within a node gadget successfully receive their messages.
\end{fact}
\begin{proof}
Note that all transmissions involving a node in a node gadget are of length 1, and therefore have a signal of either $1$ or $2$, depending on the power level of the transmission.
The claim now follows from Fact~\ref{fact:allnode} for messages sent at power level 1, and Fact~\ref{fact:boundary} for messages sent at power level 2. Therefore, all nodes within a node gadget successfully receive their messages.
\end{proof}

\begin{fact}
All nodes within an edge gadget successfully receive their messages.
\end{fact}
\begin{proof}
Note that a node transmitting near a corner of an edge gadget, as is the transmission to node $x$ shown in Figure~\ref{fig:edgeturn}, can bound its interference by the interference experienced by node $v$, shown in Figure \ref{fig:nodegadget}.  For the remaining edge gadget transmissions, we can simply bound the total interference received by any edge gadget node very conservatively by assuming that each edge gadget node has 4 nodes transmitting at each integer distance from distance 2 to infinity.  We assume there are 4 nodes transmitting at each distance so that we can be sure to upper bound the interference from transmissions coming from both nearby edge gadget nodes, and nearby node gadget nodes.  As such, these interferences can be bounded by $4 \cdot (\sum_{i=2}^{\infty} i^\alpha) \leq 4 \cdot (2^{-\alpha} + 1/(\alpha-1)/(3^{\alpha-1})) \leq .723$, for $\alpha=3$, which is less than the signal of received by any node in an edge gadget namely at least $(1.1)^{-3} \geq 0.75$.
\end{proof}

\end{document}